\documentclass[copyright]{eptcs}
\providecommand{\event}{DCFS 2010} % Name of the event you are submitting to
\usepackage{amsthm}

\newlength{\btlabelwidth}
\newlength{\btleftmargin}

\newenvironment{btlists*}{\begin{list}{{\rm--}}{%
\setlength{\labelwidth}{\btlabelwidth}\setlength{\leftmargin}{\btleftmargin}%
\addtolength{\leftmargin}{-1.3em}%
\setlength{\topsep}{0pt plus0.2ex}%
\setlength{\itemsep}{0ex plus0.2ex}%
\setlength{\parsep}{0pt plus0.2ex}}}{\end{list}}

\title{Accepting Hybrid Networks of Evolutionary Processors with Special Topologies and Small Communication}
\author{J\"urgen Dassow \qquad\qquad
Florin Manea\thanks{Florin Manea is on leave from {\it Faculty of Mathematics and Computer Science,
University of Bucharest}. His work is supported by the {\it Alexander-von-Humboldt Foundation}.}
\institute{Otto-von-Guericke-Universit\"at Magdeburg, Fakult\"at f\"ur Informatik \\
PSF 4120, D-39016 Magdeburg, Germany}
\email{\{dassow,manea\}@iws.cs.uni-magdeburg.de}
}
\def\titlerunning{AHNEPs with Special Topologies and Small Communication}
\def\authorrunning{J. Dassow \& F. Manea}
\begin{document}
\maketitle
\newcommand{\mathbb}{\bf}
\newcommand{\ra}{\rightarrow}
\newcommand{\lf}{\langle}
\newcommand{\spade}{\spadesuit}
\newcommand{\heart}{\heartsuit}
\newcommand{\rf}{\rangle}
\newcommand{\lc}{\prec\!\!}
\newcommand{\rc}{\!\!\succ}
\newcommand{\Lra}{\Longrightarrow}
\newcommand{\Ra}{\Rightarrow}
\newcommand{\lra}{\longrightarrow}
\newcommand{\bigo}{{\mathcal O}}
\def\eps{\varepsilon}
\def\mand{\mbox{ and }}
\def\mfor{\mbox{ for }}
\newcommand{\s}{\sigma}
\renewcommand{\t}{\tau}

\newcommand{\PI}{{\mathit PI}}
\newcommand{\FI}{{\mathit FI}}
\newcommand{\PO}{{\mathit PO}}
\newcommand{\FO}{{\mathit FO}}
\newcommand{\NTIME}{{\mathit NTIME}}
\newcommand{\Sub}{{\mathit Sub}}
\newcommand{\Del}{{\mathit Del}}
\newcommand{\Ins}{{\mathit Ins}}
\newcommand{\alphp}{{\mathit alph}}

\newtheorem{theorem}{\bf Theorem}
\newtheorem{corollary}{\bf Corollary}
\begin{abstract}
Starting from the fact that complete Accepting Hybrid Networks of Evolutionary Processors allow much communication
between the nodes and are far from network structures used in practice, we propose in this paper three network
topologies that restrict the communication: star networks, ring networks, and grid networks. We show that ring-AHNEPs
can simulate $2$-tag systems, thus we deduce the existence of a universal ring-AHNEP. For star networks or grid networks, we show a more general result; that is, each recursively enumerable language can be
accepted efficiently by a star- or grid-AHNEP. We also present bounds for the size of these star and grid
networks. As a consequence we get that each recursively enumerable can be accepted by networks with at most 13
communication channels and by networks where each node communicates with at most three other
nodes.\end{abstract}

\section{Introduction}
Motivated by parallel architectures, networks of language processors were introduced by E.~Csuhaj-Varj\'u and
A.~Salomaa in \cite{CsuSal97}. They consist of processors which process a word according to some rewriting rules and
communicate by sending the words to other nodes via filters. Starting from finite sets of word in the nodes a
language is generated consisting of all nodes occurring in some distinguished node.
In \cite{juan_acta} networks with evolutionary processors were introduced where the rewriting rules are restricted to
substitutions, deletions and insertions of letters which describe the local mutations known from
biology. It was shown by them that networks with evolutionary networks can generate all recursively enumerable
languages.

In \cite{dna10}, an accepting counterpart for networks with evolutionary processors was introduced. A distinguished
node contains a word which is processed by the network and accepted if a word occurs in another
distinguished node. It was shown that any recursively enumerable language can be accepted by accepting networks with
evolutionary processors (abbreviated by AHNEP). Moreover, AHENPs with seven nodes are sufficient to accept any
recursively enumerable language (see \cite{artiom09}), while AHNEPs with ten nodes are sufficient to simulate
efficiently a nondeterministic Turing machine (see \cite{actaLMM}).

In most of the papers in the area, AHNEPs with an underlying complete graph were considered. As far as universality
and computational and descriptional complexity results are concerned, it seems a natural decision to keep the shape
of the network the same, no matter what language the network accepts; otherwise the different characterizations that
one may obtain can be considered as not being uniform.
However, complete graphs have two disadvantages.
\begin{btlists*}
\item Complete graphs are very different from network structures which occur in practice where very often
stars, grids, rings, backbones etc. are used.
\item  A complete graph has many edges which correspond to communication channels. If one is interested
in a small amount of communication, one has to consider other structures for the underlying graphs.
\end{btlists*}

In this paper we consider AHNEPs where the underlying graph is a star or a grid or a ring. Regarding these structures
one can expect to get also networks with a small number of communication channels, since the number of edges in a
star is its number of nodes minus 1 and the maximal number of edges connected with a node is restricted by 4 in a
grid.

We show that each recursively enumerable language can be accepted efficiently by a star or grid-AHNEP. Furthermore, we
prove
that the number of processors used in these star networks is bounded by 13 and that of grid networks by 52.
As a consequence we get that each recursively enumerable language can be accepted by networks with at most 12 edges
and by networks where each node communicates with at most three other nodes. In the case of ring networks, we only
have a weaker result. Ring-AHNEPs can simulate $2$-tag systems, thus we deduce the existence of a universal ring-AHNEP.
Moreover, ring-AHNEPs are able to solve NP-complete problems in polynomial time.

\section{Basic definitions}
We start by summarizing the notions used throughout the paper; for all unexplained notions the reader is referred to
\cite{handbook,hartmanis2,Papa,rogo}.

An alphabet is a finite and nonempty set of symbols. The set
of all words over $V$ is denoted by $V^*$ and the empty word is denoted by $\eps$. The length of a word $x$ is
denoted by $|x|$ while $\alphp(x)$ denotes the minimal alphabet $W$ such that $x\in W^*$.
For a word $x\in W^*$, $x^r$ denotes the reversal of the word.

We say that a rule $a\ra b$, with $a,b\in V\cup\{\eps\}$ is a {\it substitution rule} if both $a$ and $b$ are not $\eps$; it is a {\it deletion rule}
if $a\ne\eps$ and $b=\eps$; it is an {\it insertion rule} if $a=\eps$ and $b\ne\eps$. The set of all substitution, deletion, and insertion rules over
an alphabet $V$ are denoted by $\Sub_V$, $\Del_V$, and $\Ins_V$, respectively.

Given a rule  $\s$ as above and a word $w\in V^*$, we define the following \emph{actions} of $\s$ on $w$:
{\begin{itemize}
\item If $\s\equiv a\ra
b\in \Sub_V$, then $\s^*(w)=\left\{
\begin{array}{ll}
\{ubv:\ \exists u,v\in V^*\ (w=uav)\},\\
\{w\mid w\mbox{ does not contain }a\},
\end{array}\right.$
\item If $\s\equiv a\to \eps\in \Del_V$, then $ \s^*(w)=\left\{
\begin{array}{ll}
\{uv:\ \exists u,v\in V^*\ (w=uav)\},\\
\{w\mid w\mbox{ does not contain }a\},
\end{array}\right.$\\
$\s^r(w)=\left\{
\begin{array}{ll}
\{u:\ w=ua\},\\
\{w\mid w\mbox{ contains no }a\},
\end{array}\right.$ \
$\s^l(w)=\left\{
\begin{array}{ll}
\{v:\ w=av\},\\
\{w\mid w\mbox{ contains no }a\},
\end{array}\right.$
\item If $\s\equiv \eps\to a\in \Ins_V$, then\\
$\s^*(w)=\{uav:\ \exists u,v\in V^*\ (w=uv)\},\ \s^r(w)=\{wa\},\ \s^l(w)=\{aw\}.$
\end{itemize}}
We say that $\alpha\in\{*,l,r\}$ expresses the way of applying a
deletion or insertion rule to a word, namely at any position
($\alpha=*$), in the left ($\alpha=l$), or in the right
($\alpha=r$) end of the word, respectively. For every rule $\s$,
action $\alpha\in \{*,l,r\}$, and $L\subseteq V^*$, we define the
\emph{$\alpha$-action of $\s$ on $L$} by
\[\s^\alpha(L)=\bigcup_{w\in L} \s^\alpha(w).\]
Given a
finite set of rules $M$, we define the \emph{$\alpha$-action of
$M$} on a word $w$ and a language $L$ by \\
\[M^{\alpha}(w)=\bigcup_{\s\in M} \s^{\alpha}(w)\
\mbox{ and } \ M^{\alpha}(L)=\bigcup_{w\in L}M^{\alpha}(w),\]
respectively. In what follows, we shall refer to the rewriting
operations defined above as {\it evolutionary operations} since
they may be viewed as linguistic formulations of local DNA
mutations.

For two disjoint subsets $P$ and $F$ of an alphabet $V$ and
a word $w$ over $V$, we define the predicates \\
\begin{eqnarray*}
\varphi^{(s)}(w;P,F)&\equiv& P\subseteq \alphp(w) \wedge F\cap \alphp(w)=\emptyset, \\
\varphi^{(w)}(w;P,F)&\equiv& (P=\emptyset \vee \alphp(w)\cap P \ne \emptyset) \wedge F\cap \alphp(w)=\emptyset .
\end{eqnarray*}

The construction of these predicates is based on {\it random-context conditions} defined by the two sets $P$
{\it permitting contexts/symbols} and $F$ of {\it forbidding contexts/symbols}. Informally, the first condition
requires that all permitting symbols are present in $w$ and no forbidding symbol is present in $w$,
while the second one is a weaker variant of the first, requiring that -- for an non-empty set $P$ -- at least one
permitting symbol appears in $w$ and no forbidding symbol is present in $w$.

For every language $L\subseteq V^*$ and $\beta\in \{(s),(w)\}$, we define:\\
$$\varphi^\beta(L,P,F)=\{w\in L\mid \varphi^\beta(w;P,F)\}.$$

An \emph{evolutionary processor over $V$} is a tuple $(M,\PI,\FI,\PO,\FO)$, where: \\
-- $(M\subseteq \Sub_V)$ or $(M\subseteq \Del_V)$ or $(M\subseteq \Ins_V)$ and \\
-- $\PI,\FI,\PO,\FO \subseteq V$. \\
$M$ is a set of substitution, deletion or insertion rules over the alphabet $V$. The set $M$ represents the set of
evolutionary rules of the processor. As one can see, a processor is ``specialized'' in only one evolutionary operation.
$\PI$ and $\FI$ are the {\it input} permitting/forbidding contexts of the processor, while $\PO$ and $\FO$ are the
{\it output} permitting/forbidding contexts of the processor. Informally, the permitting input/output contexts are
the set of symbols that should be present in a word, when it enters/leaves the processor, while the forbidding
contexts are the set of symbols that should not be present in a word in order to enter/leave the processor.

We denote the set of evolutionary processors over $V$ by $EP_V$. Note that the evolutionary processor described here is a mathematical concept
similar to that of an evolutionary algorithm, both being inspired from Darwinian evolution. The rewriting operations we have considered might be
interpreted as mutations and the filtering process described above might be viewed as a selection process. Recombination is missing but it was
asserted that evolutionary and functional relationships between genes can be captured by taking only local mutations into consideration \cite{s}.
Furthermore, we are not concerned here with a possible biological implementation of these processors, though this is a matter of great importance.

An \emph{accepting hybrid network of evolutionary processors} (AHNEP for short) is a $7$-tuple
$\Gamma=(V,$ $U,$ $G,$ ${\mathcal N},$ $\alpha,$ $\beta,$ $x_I,$ $x_O)$, where
\begin{btlists*}
\item $V$ and $U$ are the input and network alphabets, respectively, with $V\subseteq U$.
\item $G=(X_G,E_G)$ is an undirected graph, with the set of nodes $X_G$ and the set of edges $E_G$.
\item $\mathcal{N}:X_G\lra EP_U$ is a mapping which associates with each node $x\in X_G$ an evolutionary
processor $\mathcal{N}(x)=(M_x,\PI_x,\FI_x,\PO_x,\FO_x)$.
\item $\alpha: X_G\lra \{*,l,r\}$; $\alpha(x)$ gives the action mode of the rules of node $x$
on the words existing in that node.
\item $\beta: X_G\lra \{(s),(w)\}$ defines the type of the {\it input/output filters} of a node.
\item $x_I$ and $x_O\in
X_G$ is the {\it input node}, and the {\it output node}, respectively, of the AHNEP.
\end{btlists*}
$G$ is called the {\it underlying graph} of the network. We denote an undirected edge between nodes $x$ and $y$ by $(x,y)$.

For every node $x\in X_G$, the following filters are defined:
\begin{eqnarray*}
\mbox{input filter: }&&  \rho_x(\cdot)=\varphi^{\beta(x)}(\cdot;\PI_x,\FI_x), \\
\mbox{output filter: } &&  \tau_x(\cdot)=\varphi^{\beta(x)}(\cdot;\PO_x,\FO_x).
\end{eqnarray*}
That is, $\rho_x(w)$ (resp. $\tau_x$) indicates whether or not the word $w$ can pass the input (resp. output) filter
of $x$. More generally, $\rho_x(L)$ (resp. $\tau_x(L)$) is the set of words of $L$ that can pass the input
(resp. output) filter of $x$.

We say that the number of nodes present in the AHNEP $\Gamma $ is the size of $\Gamma$.

A \emph{configuration} of an AHNEP $\Gamma$ as above is a mapping $C: X_G\lra 2^{V^*}$ which associates a set of
words with every node of the graph. A configuration may be understood as the sets of words which are present in any
node at a given moment. A configuration can change either by an evolutionary step or by a communication step.

When changing by an evolutionary step each component $C(x)$ of the configuration $C$ is changed in accordance with
the set of evolutionary rules $M_x$ associated with the node $x$ and the way $\alpha(x)$ of applying these rules.
Formally, we say that the configuration $C'$ is obtained in an \emph{evolutionary step} from the configuration $C$,
written as $C\Lra C'$, if and only if $C'(x)=M^{\alpha(x)}_x(C(x))$ for all $x\in X_G$.

When changing by a communication step, each node processor $x\in X_G$ sends one copy of each word it has, which is
able to pass the output filter of $x$, to all the node processors connected to $x$ and receives all the words sent by
any node processor connected with $x$ providing that they can pass its input filter. Formally, we say that the
configuration $C'$ is obtained in a \emph{communication step} from configuration $C$, written as
$C\vdash C'$, if and only if
$$C'(x)=(C(x)-\tau_x(C(x)))\ \cup \bigcup_{(x,y)\in E_G} (\tau_y(C(y))\cap \rho_x(C(y)))$$
for all $x\in X_G$. Note that words which leave a node $x$ are eliminated from that node; if they cannot pass the
input filter of any node connected with $x$, they are lost.

Let $\Gamma$ be an AHNEP, the computation  of $\Gamma$ on the input word $w\in V^*$ is a sequence of configurations
$C_0^{(w)},C_1^{(w)},C_2^{(w)},\dots$, where the initial configuration $C_0^{(w)}$ of $\Gamma$ is defined by
$C_0^{(w)}(x_I)=\{w\}$ and $C_0^{(w)}(x)=\emptyset$ for all $x\in X_G$, $x\ne x_I$, and
$C_{2i}^{(w)}\Lra C_{2i+1}^{(w)}$ and $C_{2i+1}^{(w)}\vdash C_{2i+2}^{(w)}$ holds for all $i\geq 0$. By the previous
definitions, each configuration $C_i^{(w)}$ is uniquely determined by the configuration $C_{i-1}^{(w)}$. Thus each
computation in an AHNEP is deterministic. A computation as above immediately halts if one of the following two
conditions holds:\\
\phantom{i}(i) There exists a configuration in which the set of words existing in the output node $x_O$ is non-empty.
In this case, the computation is said to be an {\it accepting computation}.\\
(ii) There exist two identical configurations obtained either in consecutive evolutionary steps or in consecutive
communication steps. If this condition is fulfilled, but condition (i) is not fulfilled, then we say that the computation is a {\it rejecting computation}.

In the aforementioned cases the computation is said to be finite.

The {\it language accepted} by $\Gamma$ is
$$L_a(\Gamma)=\{w\in V^*\mid \mbox{ the computation of $\Gamma$ on $w$ is an accepting one}\}.$$

We say that an AHNEP $\Gamma$ decides the language $L\subseteq V^*$, and write $L(\Gamma)=L$ if and only if
$L_a(\Gamma)=L$ and the computation of $\Gamma$ halts on every $x\in V^*$.

Let $\Gamma$ be an AHNEP deciding the language $L$. The {\it time complexity} of the finite computation $C_0^{(x)}$, $C_1^{(x)}$, $C_2^{(x)}$, $\dots
C_m^{(x)}$ of $\Gamma$ on $x\in L$  is denoted by $Time_{\Gamma}(x)$ and equals $m$.
The time complexity of $\Gamma$ is the function from {\bf N} to {\bf N},\\
\centerline{$Time_{\Gamma}(n)=\mbox{sup}\{Time_{\Gamma}(x)\mid |x|=n\}.$}

In this paper we consider networks of evolutionary processors with special underlying graphs.

The first network topology we consider will be called {\it star}. An $n$-star, with $n\geq 1$, is a graph having
$n+1$ vertices, say $x_1,x_2,\ldots,x_n$ and $C$, and the edges $(x_i,C)$ for $1\leq i\leq n$. One can easily imagine
such a network as a center node $C$ connected to $n$ other nodes that are placed around it.
An AHNEP having as underlying graph a $n$-star for some $n\geq 1$, will be called $\star$AHNEP (or star-AHNEP).

The second network topology we consider will be called {\it ring}. An $n$-ring, with $n\geq 2$, is a graph having
$n+1$ nodes, say $x_1,x_2,\ldots,x_n$ (called the ring-nodes) and $x_0$ (called the center-node), and the edges
$(x_i,x_{(i \mbox{ mod }n) +1})$ and $(x_i,x_0)$ for $1\leq i\leq n$. One can easily imagine such a network as
$n$ nodes disposed on a ring, each two consecutive nodes being connected, with another node right in the middle of
the ring, connected to all the others (however, this node will have a different use than the central node in
$\star$AHNEPs). An AHNEP having as underlying graph such a ring and verifying the property that the central node $x_0$
of the graph is always the output node of the network will be called a $\circ$AHNEP (or ring-AHNEP). In this way, the
computation is done by the ring-nodes, and it stops only when a string enters the center-node (who does not have any
other effect on the computation). It is worth to remark that the topology we use in this case is called ring because
the nodes that perform the computation form a ring; putting the output node on the same ring would lead to the
disconnection of this ring and its transformation into a chain (because the output node cannot be of any use in the
computation, and, moreover, once a string enters in this node, the whole process stops).

The third network topology we consider will be called a grid. An $(m,n)$-grid is a graph consisting of $m\cdot n$
nodes $(i,j)$, $1\leq i\leq m$ and $1\leq j\leq n$, and having the edges
$((i,j),(i+1,j))$ and $((i,j),(i,j+1))$ for $1\leq i\leq m-1$ and $1\leq j\leq n-1$,
$((m,j),(m,j+1))$ for $1\leq j\leq n-1$ and $((i,n),(i+1,n))$ for $1\leq i\leq m-1$. Obviously, an
$(m,n)$-grid is isomorphic to the finite subgrid of the grid $\mathbb{Z}\times \mathbb{Z}$ of height $n-1$ and width
$m-1$. An AHNEP having  as an underlying graph an $(m,n)$-grid for some $n\geq 1$ and $m\geq 1$ will be called a
grid-AHNEP.

Ring and star networks were also approached in \cite{juan_acta} but in the context of generating NEPs with regular
filters. However, the results and proofs presented in that paper cannot be applied to get similar results for AHNEPs
with random context filters.

\section{Star-AHNEPs}
We start with a further characterization of the class of recursively enumerable (recursive) languages by networks where the underlying graphs are
complete graphs with loops.
\begin{theorem}\label{th-loops}
For any recursively enumerable (recursive) language $L$ accepted (decided) by a Turing machine there exists a
AHNEP $\Gamma$ of size $12$ such that the underlying graph is a complete graph which additionally has a loop for each
node, and it accepts (decides) $L$. \\
Moreover, if $L\in \NTIME(f(n))$, then $Time_\Gamma(n)\in {\cal O}(f(n))$ (where the constant hidden by the $\bigo$ notation depends on the Turing machine accepting, respectively deciding, $L$).
\end{theorem}
By limitedness of space we omit the proof. Note, however, that it follows mainly the proof of the main result in \cite{actaLMM}.

The main result that we show in this section is the following theorem.
\begin{theorem}\label{starcompl}
For any recursively enumerable (recursive) language $L$ accepted (decided) by a Turing machine there exists a
$\star$AHNEP $\Gamma$ of size $13$ accepting (deciding) $L$. Moreover, if $L\in \NTIME(f(n))$, then
$Time_\Gamma(n)\in {\cal O}(f(n))$ (where the constant hidden by the $\bigo$ notation depends on the Turing machine accepting, respectively deciding, $L$).
\end{theorem}
\begin{proof}
The proof is based on the following simple remark. Any AHNEP having as underlying graph a complete graph with loops (i.e., there is an edge between
every two nodes, and an edge connecting each node to itself) can be canonically transformed into a $\star$AHNEP. Assume that the initial AHNEP has
$n$ nodes. In the new $\star$AHNEP we will have all the nodes of the initial network (with the same set of rules and filters), and a new one, called
here $C$ (with no rules and filters allowing any word to pass the filters). Every node of the former network is connected to $C$, and no other edges exist. It is not hard to see that
the new network accepts (decides) exactly the same language that is accepted (decided) by the initial one. Indeed, the evolution steps are performed
by the same nodes in both networks, but the communication that takes place between two nodes of the initial network is done here via the $C$ node
(that does not modify in any way the string). In conclusion, we have a way to transform AHENPs that have as underlying network a complete graph with
loops into $\star$AHNEPs, by adding only one node. Moreover, the $\star$AHENP that we obtain simulates efficiently a computation of the initial
AHENP: it makes at most two times the number of steps as in the initial computation of the complete AHNEP with loops.

Now the result follows immediately from Theorem \ref{th-loops}.
\end{proof}

Using the same idea we can prove the following statement.
\begin{theorem}\label{2tag}
For every $2$-tag system $T=(V,\phi)$ there exists a $\star$AHNEP $\Gamma$ of size $7$ such that $L(\Gamma)=\{w\mid \mbox{ $T$ halts on $w$}\}$.
\qed
\end{theorem}

We stress out that this result does not improve, neither from the descriptive complexity point of view nor from the
computational complexity point of view, the result in Theorem \ref{starcompl}. Indeed, not every recursively
enumerable language can be accepted by $2$-tag systems, thus Theorem \ref{2tag} does not imply that $\star$AHNEPs
with $7$ nodes are complete; moreover, although $2$-tag systems efficiently simulate deterministic Turing machines,
via cyclic tag systems (see, e.g., \cite{woods}), the previous result does not allow us to infer a bound on the size
of the networks accepting in a computationally efficient way all recursively enumerable languages.

However, from Theorem \ref{2tag}, and the fact that $2$-tag systems are universal, we derive the following fact.
\begin{corollary} There exists a universal $\star$AHNEP with $7$ nodes. \qed
\end{corollary}

In this respect we discuss a bit on the importance of the descriptive complexity bounds shown here: we know that
$\star$AHNEPs with $13$ nodes accept efficiently recursive enumerable languages, and, as a simple consequence, we
have that NP equals the class of languages accepted in polynomial time by $\star$AHNEPs with $13$~nodes.
Why is such a result interesting, since, after all it is a rather clear fact that AHNEPs simulate the nondeterminism
(of Turing machines, for instance) by working with strings with multiplicities? Our answer to this question would
be that the result is interesting because it shows that a special class of AHNEPs (with a constant number of nodes, and
a special topology) has these good computational properties; it is not clear if AHNEPs with less nodes have the same
properties (we were not able neither to get a lower bound neither to show an equivalent of Theorem \ref{starcompl}
for smaller AHNEPs).

If we are interested in a small number of communication channels in the net, then we have the following result which
follows immediately from Theorem \ref{starcompl}.

\begin{corollary}
For any recursively enumerable language $L$, there is an AHNEP $\Gamma$ with at most 12 edges in the underlying graph
such that $L=L_a(\Gamma)$. \qed
\end{corollary}

\section{Grid-AHNEPS}

In this section we prove that grid-AHNEPS have the same power as $\star$AHNEPs. More precisely, we show the following
the following statements.

\begin{theorem}\label{th-grid}
Let $\Gamma =(V,U,G,{\mathcal N},\alpha,\beta,x_i,x_o)$ be a network of size $n$ where the underlying graph $G=(X_G,E_G)$ is a complete graph with
loops. Then there is a $(4,n+1)$-grid-AHNEP $\Gamma'$ such that $L_a(\Gamma')=L_a(\Gamma)$.
\end{theorem}
\begin{proof}
Let $\Gamma =(V,U,G,{\mathcal N},\alpha,\beta,x_i,x_o)$ be a network where the underlying graph
$G=(\{1,2,\dots ,n\},E_G)$ is a complete graph with loops $(x,x)$ for any node $x\in X_G$. Then we construct the
grid-AHNEP
\[\Gamma'= (V',U,G',{\mathcal N}',\alpha',\beta',(3,x_i+1),(3,x_o+1))\]
where $G'$ is the grid consisting of the nodes $(i,j)$ with $1\leq i\leq n+1$, $1\leq j\leq 5$,
\begin{eqnarray*}
&& V' = V\cup \bigcup_{i=0}^n \{ X_i,X_i'Y_i,Y_i'\} ,  \\
&&{\mathcal N}'(1,1)=(\{X_0\ra X_1\}, \{X_0\}, \emptyset, \emptyset, \emptyset), \ \alpha'(1,1)= *,\ \beta'(1,1)= (w),
 \\
&& {\mathcal N}'(1,i+1)=(\{X_i\ra X_{i+1}, X_i\ra X_i'\}, \{X_i\}, \emptyset, \emptyset, \emptyset),\
\alpha'(1,i+1)= *,  \\
&& \qquad \qquad \beta'(1,i+1)= (w) \mfor 1\leq i\leq n-1,  \\
&&{\mathcal N}'(1,n+1)=(X_n\ra X_n'\}, \{X_n\}, \emptyset, \emptyset, \emptyset), \ \alpha'(1,n+1)= *, \
\beta'(1,n+1)= (w),  \\
&&{\mathcal N}'(2,1)=(\{Y\ra \varepsilon\}, \{X_0\}, \emptyset, \emptyset, \{Y\}), \ \alpha'(2,1)= * ,
\beta'(2,1)= (w), \\
&&{\mathcal N}'(2,i+1)=(\{ X_i'\ra \varepsilon \}, \{X_i'\}, \emptyset, \emptyset, \emptyset), \ \alpha'(2,i+1)= *,
\beta'(2,i+1)= (w) \\
&& \qquad \qquad \mfor 1\leq i\leq n-1,  \\
&&{\mathcal N}'(3,1)=(\{Y' \ra X_0\}, \{Y'\}, \emptyset, \emptyset, \emptyset), \ \alpha'(3,1)= *,
\ \beta'(3,1)= (w), \\
&& {\mathcal N}'(3,i+1)=(M_i, \PI_i, \FI_i\cup \{ Y,X_0\}, \PO_i, \FO_i), \ \alpha'(3,i+1)= \alpha(i),  \\
&& \qquad \qquad \beta'(3,i+1)= \beta(i) \mfor 1\leq i\leq n,  \\
&&{\mathcal N}'(4,1)=(\{Y\ra Y'\}, \{Y\}, \emptyset, \emptyset, \emptyset),\ \alpha'(4,1)= *, \ \beta'(4,1)= (w) \smallskip \\
&&{\mathcal N}'(4,i+1)=(\{ \varepsilon \ra Y\}, \emptyset, \{Y'\}, \emptyset, \emptyset),
\ \alpha'(4,i+1)= *,  \\
&& \qquad \qquad \beta'(4,i+1)= (w) \mfor 1\leq i\leq n,  \\
&&{\mathcal N}'(4,n+1)=(\{ \varepsilon \ra Y_n'\}, \emptyset, \{Y_{n-1}'\}, \emptyset, \emptyset), \ \alpha'(4,n+1)= *, \
\beta'(4,n+1)= (w)
\end{eqnarray*}

We now consider the case that a word $z\in V^*$ is contained in the node $(3,i+1)$ for some $i$, $1\leq i\leq n$.
Note that in the initial configuration the word $w$ is in $(3,x_i+1)$. Then we apply an evolutionary rule of $M_i$
and obtain a word $z'\in V^*$. This can be communicated to node $(3,i)$ and $(3,i+2)$ (if $i\geq 2$ and $i\leq n-1$,
respectively), if the corresponding filters allow the communication. Since $G$ is a complete graph, these
communications can be done in $\Gamma$, too. Furthermore, $z'$ can be communicated to $(4,i+1)$. The application of
$\varepsilon\ra Y$ gives a word $z_1'Yz_2'$ with $z'=z_1'z_2'$. Now the word can only move in the fourth row of nodes,
i.\,e., only to nodes $(4,j)$ with $1\leq j\leq n+1$. As long as nodes $(4,j)$ with $2\leq j\leq n+1$ are entered,
a further $Y$ is inserted somewhere in the word. If $(4,1)$ is entered, then one $Y$ is changed to $Y'$ and the word
has to move to node $(3,1)$ where $Y'$ is replaced by $X_0$. Now the word has to move to node $(2,1)$ where the
remaining $Y$s in the word are deleted. The obtained word $z_1''X_0z_2''$ satisfies $z_1''z_2''=z'$ and has to enter
$(1,1)$. By the rule set associated with $(1,1)$, $X_0$ is replaced by $X_1$ and the word $z_1''X_1z_2''$ is sent
to $(1,2)$. Now using $X_k\ra X_{k+1}$ in a node $(1,k+1)$ we move the word in horizontal
direction to $(1,k+2)$, whereas the application of $X_i\ra X_i'$ results in a vertical move to $(2,k+1)$.
One of these possibilities have to be chosen in each node $(1,k+1)$. Let us assume that we choose the vertical
direction for some $j$. Then we obtain $z_1''X_j'z_2''$ in node $(2,j+1)$. The application of $X_j'\ra \varepsilon$
leads to $z_1''z_2''=z'$ in node $(3,j+1)$, if the filters of $(3,j+1)$ allow the passing. Altogether we have performed
a move from $(i,3)$ to $(j,3)$ and taking into consideration the filters the move from $i$ to $j$ can be done in
$\Gamma$, too. Therefore we have $L_a(\Gamma')\subseteq L_a(\Gamma)$.

Let us denote by $s_{i,j}$ the sequence of nodes
\[(4,i+1), (4,i), (4,i-1), \dots , (4,1), (3,1), (2,1), (1,1), (1,2), \dots (1,j+1), (2,j+1).\]
Let $w\in L_a(\Gamma)$. Then there is a sequence $x_i,t_1,t_2,\dots ,t_n,x_o$ of nodes which are entered in this
succession in the computation on $w$ starting in $x_i$. Then we have in $\Gamma'$ a computation on $w$ which follows
the sequence
\[(3,x_i+1),s_{x_i,t_1},(3,t_1+1),s_{t_1,t_2},(3,t_2+1),s_{t_2,t_3} \dots (3,t_n+1)s_{t_n,x_o}(3,x_o+1).\]
Thus $w\in L_a(\Gamma')$, too. Therefore we also have $L_a(\Gamma)\subseteq L_a(\Gamma')$.
\end{proof}

The following corollary follows immediately from Theorem \ref{th-loops}.
\begin{corollary}\label{c-grid}
For any recursively enumerable language $L$, there is a grid-AHNEP $\Gamma'$ of size at most 52 such that
$L=L_a(\Gamma')$. \qed
\end{corollary}

Obviously, any node of a grid-AHNEP is connected with at most four other nodes. Thus, in order to accept
recursively enumerable languages it is sufficient to consider graphs where each node has a degree $\leq 4$. This
result is slightly improved by the following statement.

\begin{theorem}\label{th-degree}
For any recursively enumerable language $L$, there is a AHNEP $\Gamma$ such that any node of the underlying graph
has a degree $\leq 3$ and $L_a(\Gamma)=L$.
\end{theorem}
\begin{proof}
We consider the network constructed in the proof of Theorem \ref{th-grid}. It is obvious,
that the edges $((i,j), (i+1,j))$ for $1\leq j\leq n$ and $2\leq j\leq 3$ can be omitted without changing the
accepted language. Clearly, after the cancellation the maximal degree of the nodes is three.
\end{proof}

We do not know whether or not this result is optimal. If the degree is bounded by 2, one gets chains of nodes
(see the remark following the definition of ring-AHNEPs). We leave as an open question whether AHNEPs with chains as
underlying graphs are able to accept all recursively enumerable languages.

\section{Ring-AHNEPs}
With respect to ring-AHNEPs we were not able to prove that all recursively enumerable languages can be accepted.
However, ring-AHNEPs also have a large power since they are able to simulate the universal $2$-tag systems and
to solve NP-complete problems in polynomial time.

\begin{theorem}\label{ring-tag}
For every $2$-tag system $T=(V,\phi)$ there exists a $\circ$AHNEP $\Gamma$ of size $9$ such that $L(\Gamma)=\{w\mid \mbox{ $T$ halts on $w$}\}$.
\end{theorem}
We omit the proof by limitedness of space.

Once again, we remark that this result does not give a characterization of what class of languages $\circ$AHNEPs accept. We only proved that for any
recursively enumerable language one can construct a $\circ$AHNEP that accepts a special encoding of that language. Since $2$-tag systems efficiently
simulate deterministic Turing machines (\cite{woods}), the previous result also shows that $\circ$AHNEPs with $8$ nodes
also accept efficiently encodings of the languages accepted efficiently by deterministic Turing machines.

A direct consequence of Theorem \ref{ring-tag} is the following assertion.
\begin{corollary} There exists a universal $\circ$AHNEP with $9$ nodes. \qed
\end{corollary}

Finally, in this section we present an application of $\circ$AHNEPs to efficient problem solving.
First, let us recall what we mean by solving a problem with AHNEPs (see \cite{mscs}). A correspondence between decision problems and languages can be done via an encoding function which transforms an instance of a given
decision problem into a word, see, e.g., \cite{17}. We say that a decision problem $P$ is solved in time $O(f(n))$ by AHNEPs if there exists a family
$\cal G$ of AHNEPs such that the following conditions are satisfied:
\begin{btlists*}
\item The encoding function of any instance $p$ of $P$ having size $n$ can be computed by a deterministic
Turing machine in time $O(f(n))$.
\item For
each instance $p$ of size $n$ of the problem one can effectively construct, in deterministic time $O(f(n))$, an
AHNEP $\Gamma(p)\in \cal G$ which
decides, again in time $O(f(n))$, the word encoding the given instance. This means that the word is decided if and only if the solution to the given instance of the problem is ``YES''. This effective construction is called an $O(f(n))$ time solution to the considered problem.
\end{btlists*}

Satisfiability is perhaps the best studied NP-complete problem because one arrives at it from a large number of practical problems. It has direct
applications in mathematical logic, artificial intelligence, VLSI engineering, computing theory, etc. It can also be met indirectly in the area of
constraint satisfaction problems. We now mention a result on getting all solutions to an instance of 3CNF-SAT (see \cite{17}) using $\circ$AHNEPs.

\begin{theorem}\label{satis}
The satisfiability of formulas given in 3-Conjunctive Normal Form can be solved in polynomial time by $\circ$AHNEPs.
\end{theorem}
The proof of this theorem is based on the following idea: for a formula over $n$ variables, we generate in
$\bigo (n^2)$ AHNEP-time all the possible assignments for these variables, in a massively parallel manner; then,
if the formula is the conjunction of $m$ clauses, we evaluate it, for all the possible assignments, in parallel,
in $\bigo(m)$ time. By limitedness of space we omit the detailed proof.
%\medskip
Also, the result of the previous theorem can be extended, using a very
similar proof, for a class of AHNEPs with simpler underlying graph,
namely AHNEPs in which every node has degree $2$ (or, simpler, AHNEPs
having as underlying graph a chain of nodes): the satisfiability of
formulas given in 3-Conjunctive Normal Form can be decided in
polynomial time by such AHNEPs.

Note that this result does not say much about the
computational power of AHNEPs having the underlying graph a chain; we
were not able to show that they can be used to accept arbitrary
recursively enumerable languages (as we did for star-AHNEPs and grid-AHNEPs, in
Theorem \ref{starcompl} and, respectively, Corollary~\ref{c-grid}), or even encodings of such languages (as we
did for ring-AHNEPs, in Theorem~\ref{ring-tag}). In fact, we only
showed that for all the instances of the Satisfiability problem,
having the same number of variables, one can construct an AHNEP that
decides for which formulas there exists an assignment of the variables
that make them true, and for which there is no such assignment; it is
important to note that a new network should be constructed every time
we have to check a formula with a different number of variables.

Besides showing the completeness of star-AHNEPs and grid-AHNEPs, Theorem \ref{starcompl} and Corollary \ref{c-grid} show that networks with star or grid topology can be used to efficiently solve NP-complete problems. However, as far as problem solving is concerned, we can easily
adapt the constructions presented in \cite{mcu,mscs} to see how NP-complete problems can be solved efficiently by
complete AHNEPs with loops; thus, we can easily obtain efficient solutions implemented on star-AHNEPs, or grid-AHNEPs, for the same problems.

\bibliographystyle{eptcs}

 % or whatever you prefer
\end{document}